\documentclass[11pt]{article}
\usepackage[a4paper,centering,scale=0.75]{geometry}
\usepackage{amsfonts,amsmath,amsthm}
\usepackage{natbib}
\usepackage{graphicx}
\usepackage{setspace}

\newtheorem{prop}{Proposition}[section]
\newtheorem{lemma}[prop]{Lemma}
\newtheorem{coroll}[prop]{Corollary}
\theoremstyle{remark}
\newtheorem{rmk}[prop]{Remark}

\newcommand{\1}{\mathbf{1}}
\newcommand{\E}{\mathop{{}\mathbb{E}}}
\newcommand{\R}{\mathbb{R}}
\newcommand{\Pb}{\mathbb{P}}
\renewcommand{\P}{\mathbb{P}}
\newcommand{\erre}{\mathbb{R}}

\DeclareMathOperator{\sgn}{sgn}
\newcommand{\ip}[2]{\langle #1,#2 \rangle}

\title{On the relation between forecast precision and trading profitability
of financial analysts}
\author{Carlo Marinelli\\
\small Department of Mathematics\\
  \small University College London, United Kingdom\\
\and 
Alex Weissensteiner\\
\small School of Economics and Management\\
  \small Free University of Bolzano, Italy\\}

\date{28 January 2013}

\begin{document}
\maketitle
\begin{abstract}
  We analyze the relation between earning forecast accuracy and
  expected profitability of financial analysts. Modeling forecast
  errors with a multivariate Gaussian distribution, a complete
  characterization of the payoff of each analyst is provided. In
  particular, closed-form expressions for the probability density
  function, for the expectation, and, more generally, for moments of
  all orders are obtained. Our analysis shows that the relationship
  between forecast precision and trading profitability need not to be
  monotonic, and that, for any analyst, the impact on his expected
  payoff of the correlation between his forecasts and those of the
  other market participants depends on the accuracy of his
  signals. Furthermore, our model accommodates a unique
  full-communication equilibrium in the sense of \cite{Radner1979}: if
  all information is reflected in the market price, then the expected
  payoff of all market participants is equal to zero.
\end{abstract}

\section{Introduction}
Many empirical studies indicate that financial analysts differ in
their forecast accuracy \citep[see, e.g.,][]{Stickel1992, Sinha1997},
and that these differences are persistent over time
\citep[see][]{mikhail2004}. Therefore it is natural to ask how the
forecast ability of an analyst translates into the profitability of a
trading strategy based on his advice. This question is addressed,
e.g., in the works of \cite{Loh2006} and \cite{Ertimur2007}. Both
papers rank analysts according to their earnings forecast accuracy and
show that the difference between factor-adjusted returns resulting
from recommendations in the highest accuracy quintile and in the
lowest accuracy quintile is significantly positive.
However, as noted by \cite{Ertimur2007},
since both papers focus only on the contemporaneous relationship
between accuracy and profitability, the reported abnormal excess
returns among analysts cannot be considered as evidence for the
existence of an implementable ex-ante trading strategy.

Furthermore, a strand of literature reports that earning abnormal
trading returns based on the recommendations of financial analysts is
by no means an easy task: \cite{Bradshaw2004} shows that although
earning forecasts have the highest explanatory power for
recommendations, these projections have the least association with
future excess returns; \cite{Barber2001} and \cite{mikhail2004}
conclude that, after trading costs are taken into consideration, the
differences in trading performance among analysts become
insignificant; \cite{Brown2008} argue that reported abnormal returns
might be spurious due to the fact that forecast errors are scaled by
share prices.

To explain the absence of a clear positive relationship between
forecast precision and trading profitability one can simply invoke the
efficient market hypothesis: if market prices reflect correctly all
available information, then, due to the level playing field, no market
participant can earn abnormal returns. The paradox of the efficient
market hypothesis is that, if every investor believed in the
efficiency of the market, then the market would not be efficient
because no one would have an incentive to process information. This
issue is addressed by \cite{Grossman1980}, who argues that the
strong-form efficient market hypothesis is not a meaningful assumption
and that gathering information makes sense up to the point where its
marginal cost equals its marginal benefit.

Another natural way to analyze the problem is to consider inefficient
markets. Different simulation studies by \cite{Schredelseker1984,
  Schredelseker2001} and \cite{Pfeifer2009} show that for markets out
of equilibrium and with asymmetric information the relationship
between forecast accuracy and trading profitability might be
non-monotonic. This could be an explanation of the fact that linear or
monotonic statistical measures, such as Pearson or Spearman
correlation, detect only weak or no dependence in the empirical
data. The main findings were confirmed also in experimental market
settings, see \cite{huber2007} and
\cite{Huber2008}. \cite{Lawrenz2012} propose a theoretical model where
agents learn, i.e. improve their forecast abilities in a Bayesian
sense. They use numerical integration to calculate the expected profit
of the single agents, and with a regression analysis they explain the
result by forecast errors of the single analyst, by forecast errors of
the others and by correlation effects among market
participants. Surprisingly, existing empirical studies in this field
seem to neglect the covariance of the forecast errors as an important
explanatory variable.

Let us now turn to the objectives of this work: we propose a
one-period model with asymmetric information, in the context of which
we derive closed-form expressions for the probability density
function, the expectation, as well as all higher moments of the payoff
obtained by each market participant. We are not aware of other
contributions in the literature where such complete characterizations
in a model with asymmetric information have been obtained.
Our analysis yields that any market participant benefits from a
reduction (increase) in the volatility of his own signal (of the
signals of other market participants) whenever his own signal is
negatively correlated with the aggregated (weighted) signal of all
other market participants. That is, the expected payoff of any agent
improves if his signal becomes more accurate, provided his signal is
negatively correlated to the aggregated signal of all other agents.
This observation can be interpreted saying that, in this case, agents
have an incentive to improve their forecast skills (on this issue cf.,
e.g., \cite{Mikhail1999}, \cite{BravHeaton02} or \cite{Markov2006}).
On the other hand, if his signal is positively correlated with the
aggregated forecasts of all other market participants, a $J$-shaped
relationship between forecast accuracy and expected payoff
exists. Such a non-monotonic relationship was first reported in
simulation studies by \cite{Schredelseker1984, Schredelseker2001} and
empirically confirmed by \cite{huber2007} and \cite{Huber2008}. Our
model provides a rigorous explanation for the emergence of such
effects. 

Another implication of our analysis is that, for any agent, the impact
of the correlation between his own signal and the aggregated signal of
all other participants on his expected payoff depends on the accuracy
of the signal: if the relative accuracy of this signal is above
(below) a certain threshold, the impact of an increase in correlation
is positive (negative), while for intermediate levels of accuracy the
impact of correlation is non-monotonic.
In this sense our model offers an explanation of two (apparently)
contradicting empirical observations that have appeared in the
literature, i.e. that analysts with a high reputation tend to issue
similar predictions \citep[known as herding effect, see,
e.g.,][]{Graham1999}, as well as to produce recommendations that deviate
significantly from the consensus forecasts \citep[see,
e.g.][]{Lamont2002}, or to follow an anti-herding strategy \citep[see,
e.g.][]{Bernhardt2006}.

Finally, our model accommodates a full communication equilibrium in
the sense of \cite{Radner1979}, and this equilibrium is unique. If all
available information is correctly reflected in the market price, then
the expected trading payoff of all analysts in our model is equal to
zero.

The remaining content is organized as follows: in Section
\ref{s:model} we introduce the model (in particular, we obtain an
expression for the expected payoff of each agent), discuss some of its
implications through an accurate sensitivity analysis, and prove the
existence and uniqueness of a full communication equilibrium. In
Section \ref{s:joint} we obtain the probability density function and
we compute moments of all orders for the payoff of each agent. A
numerical example is provided in Section \ref{s:num}, and Section
\ref{s:con} concludes.

\medskip

\noindent \textbf{Notation.} We denote the Euclidean scalar product by
$\ip{\cdot}{\cdot}$. The Gaussian law on $\R^n$ with mean $m$ and
covariance matrix $Q$ is denoted by $N_{\R^n}(m,Q)$, and we omit the
subscript if the space is clear. The distribution and density
functions of the standard Gaussian law on $\R$ are denoted by $\Phi$
and $\phi$, respectively.

\section{Model}\label{s:model}
In order to model incomplete information we assume that $n$
risk-neutral market participants do not know the true (or fair) value
of a company. However, they process available information (e.g.,
accounting statements) and try to infer the true value of the firm
\cite[see, e.g.,][]{Barron1998, Markov2006}. Although we assume that
each market participant $i$, $i=1,2,\ldots,n$, has an unbiased signal
(i.e., forecast) $\xi_i$ about the fair value, agents are
heterogeneous along two dimensions. In particular, agents differ in
the precision of their estimates. In this way we capture the idea that
analysts might have distinct skills and/or data at their
disposal. Moreover, we assume that the correlation between individual
forecasts may differ among analysts \citep[see, e.g.,][]{Fischer1998}.
We assume that $\xi=(\xi_1,\dots,\xi_n)$ is a vector of centered
jointly Gaussian random variables with non-singular covariance matrix
$Q$.

According to their price forecasts, analysts submit conditional buy
and sell orders. Whenever the price is below his own estimate
($p<\xi_i$), analyst $i$ is willing to buy and therefore takes a long
position, otherwise ($p>\xi_i$) he takes a short position. The true
value is revealed immediately after the trade. Following
\cite{Fischer1998}, we consider a dealer market, i.e. we assume that a
market maker clears the market by buying and selling for his own
account. We assume that this market maker sets the price $p$ equal to
a weighted average of the different estimates, that is,
\begin{equation}
p := \sum_{i=1}^n w_i\xi_i = \ip{w}{\xi},
\label{eq:p}
\end{equation}
where $w=(w_1,\ldots,w_n) \in \R^n$, $\sum_i w_i=1$. In order to avoid
degenerate (and trivial) cases, we assume that at least two elements
of $w$ are not zero. Of course, setting the price equal to the mean of
the single estimates is a special case of this pricing mechanism. If
the market maker exploited additional information, then he could
assign more weight to more accurate analysts. This is in line with
empirical studies which report higher price reactions to forecast
revisions of analysts with a higher reputation, see, e.g.,
\cite{Gleason2003}. Furthermore, as we will show in
Subsection~\ref{s:ree}, the linear pricing function \eqref{eq:p} is
flexible enough to incorporate correctly all signals revealed by the
market participants and to lead to a full communication equilibrium in
the sense of \cite{Radner1979}.

We assume, without loss of generality, that the true value of the
asset is zero, and, in analogy to \cite{BravHeaton02}, we calculate
the expected performance of all agents. For notational convenience we
focus on agent 1, but it is clear that the extension of our analysis
to the generic agent $i$, $i\in\{1,\ldots,n\}$, is just a matter of
relabeling.  The expected trading payoff of agent 1 is given by
\[
  \mu_1 := \E\bigl[ (0-p)\sgn(\xi_1-p) \bigr]
  = -\E\bigl[ p\,\sgn(\xi_1-p) \bigr],
\]
where $\sgn: x \mapsto 1_{\left]0,+\infty\right[}(x) -
1_{\left]-\infty,0\right[}(x)$ stands for the signum function. In
particular, if the forecast $\xi_1$ is above the market price $p$,
then agent $1$ takes a long position with a payoff equal to $(0-p)$,
otherwise, if his forecast $\xi_1$ is below the market price $p$, his
payoff is equal to $(p-0)$.

\medskip

As a first step, we determine the joint distribution of $\xi_1$ and
$p$.
\begin{lemma}     \label{lm:S}
  Let $Q_1$ denote the first row of the matrix $Q$. One has $(\xi_1,p)
  \sim N_{\erre^2}(0,S)$, where
  \[
  S =
  \begin{bmatrix}
    q_{11} & \ip{Q_1}{w}\\
    \ip{Q_1}{w} & \ip{Qw}{w}
  \end{bmatrix}
  \]
  and $\det S>0$. In particular, one has $(\xi_1,p) = (aX+bY,cY)$ in
  law, where $(X,Y) \sim N_{\erre^2}(0,I)$ and
  \begin{equation*}
  a := \sqrt{q_{11}-b^2} >0, \qquad 
  b := \frac{\ip{Q_1}{w}}{\sqrt{\ip{Qw}{w}}}, \qquad
  c := \sqrt{\ip{Qw}{w}} >0.
  \end{equation*}
\end{lemma}
\begin{proof}
  Since $(\xi_1,p) = A\xi$, where $A: \R^n \to \R^2$ is the linear map
  represented by the matrix (which we denote by the same symbol, with
  an innocuous abuse of notation)
  \begin{equation}     \label{eq:map}
  A=
  \begin{bmatrix}
  1 & 0 & \cdots & 0 \\
  w_1 & w_2 & \cdots & w_n  
  \end{bmatrix},
  \end{equation}
  well-known results on Gaussian laws imply that $(\xi_1,p) \sim
  N_{\R^2}(0,S)$, where
  \begin{equation}
    S=AQA'=
    \begin{bmatrix}
        q_{11} & \sum_{j=1}^n w_j q_{1j}\\
        \sum_{i=1}^n w_i q_{i1} & \sum_{i,j=1}^n w_iw_jq_{ij}
    \end{bmatrix}
    =
    \begin{bmatrix}
      q_{11} & \ip{Q_1}{w}\\
      \ip{Q_1}{w} & \ip{Qw}{w}
    \end{bmatrix}.
    \label{eq:cov1}
  \end{equation}
  Note that, due to our assumptions on $w$, $A$ has full rank, hence
  $S$ is non-singular. In particular, there exists an upper-triangular
  matrix $B$ of the type
  \begin{equation}
    B=
    \begin{bmatrix}
      a & b\\
      0 & c
    \end{bmatrix}
  \end{equation}
  such that $S=BB'$. Elementary computations show that one has
  \[
  a = \sqrt{q_{11}-b^2}, \qquad 
  b = \frac{\ip{Q_1}{w}}{\sqrt{\ip{Qw}{w}}}, \qquad
  c = \sqrt{\ip{Qw}{w}}.
  \]
  Note that $a$ is well defined (and strictly positive) because 
  \[
  q_{11} - b^2 = \frac{q_{11}\ip{Qw}{w}-\ip{Q_1}{w}^2}{\ip{Qw}{w}}
  = \frac{\det S}{\ip{Qw}{w}},
  \]
  where $\det S>0$ since, as remarked above, $S$ is non-singular, and
  $\ip{Qw}{w}>0$ because $Q$ is strictly positive definite and $w \neq
  0$.  Let $Z:=(X,Y) \sim N_{\R^2}(0,I)$. Then the law of the random
  vector $BZ=(aX+bY,cY)$ is $N(0,BB')=N(0,S)$, i.e. it coincides with
  the law of $(\xi_1,p)$.
\end{proof}
Writing
\[
S =
\begin{bmatrix}
  \operatorname{Var} \xi_1 & \operatorname{Cov}(\xi_1,p)\\
\operatorname{Cov}(\xi_1,p) & \operatorname{Var} p 
\end{bmatrix}
=:
\begin{bmatrix}
  \sigma^2_1 & \sigma_{1p}\\
  \sigma_{1p} & \sigma^2_p
\end{bmatrix},
\]
it is immediately seen that the following identities hold:
\[
a = \sqrt{\sigma^2_1-\left(\frac{\sigma_{1p}}{\sigma_p}\right)^2},
\qquad
b = \frac{\sigma_{1p}}{\sigma_p},
\qquad 
c = \sigma_p.
\]
As the main result of this section, in the following we derive a
closed-form expression for the expected payoff of the single agent.
\begin{prop}     \label{prop:22}
  Let $a$, $b$, $c$ be defined as in Lemma \ref{lm:S}, and
  $\beta:=(c-b)/a$. One has
  \[
  \mu_1 = -\E\bigl[ p \sgn(\xi_1-p) \bigr] =
  \frac{2c\beta}{\sqrt{2\pi} \sqrt{1+\beta^2}}.
  \]
\end{prop}
\begin{proof}
  Taking into account that $a>0$, we can write
  \begin{equation*}
    \mu_1=-\E \bigl[ p\,\sgn(\xi_1-p) \bigr] 
    = -c\E \bigl[ Y \sgn(aX+(b-c)Y) \bigr]
    = -c \E \bigl[ Y \sgn(X - \beta Y) \bigr].
  \end{equation*}
  The independence of $X$ and $Y$ yields
  \begin{equation*}
    \E\bigl[ Y \sgn(X - \beta Y) \bigr] =
    \E\E\bigl[Y \sgn(X - \beta Y) \big| Y \bigr] =
    \int_\erre y \E\bigl[\sgn(X - \beta y)\bigr] \phi(y)\,dy.
  \end{equation*}
  Observing that one has
  \begin{align*}
    \E\bigl[\sgn(X - \beta y)\bigr] &= 
    \E 1_{\{X>\beta y\}} - \E 1_{\{X<\beta y\}}
    = \Pb(X>\beta y) - \Pb(X<\beta y)\\
    &= 1 - 2\Pb(X<\beta y) = 1 - 2\Phi(\beta y),
  \end{align*}
  we get, recalling that $\int_\R y\phi(y)\,dy=0$,
  \begin{equation*}
    \E\bigl[ Y \sgn(X - \beta Y) \bigr] =
    \int_\R y\bigl( 1-2\Phi(\beta y) \bigr) \phi(y)\,dy =
    -2 \int_\R \Phi(\beta y) \, y \, \phi(y)\,dy.
  \end{equation*}
  We thus have
  \begin{equation*}
    \mu_1 = -c \E\bigl[ Y \sgn(X-\beta Y) \bigr] =
    2c \int_\R \Phi(\beta y) y \phi(y)\,dy.
  \end{equation*}
  Since $\phi'(y)=-y\phi(y)$ and $\Phi'(\beta y)=\beta \phi(\beta y)$,
  integration by parts yields, taking into account that $\Phi$ is
  bounded and $\phi$ is rapidly decreasing at infinity,
  \[
  \mu_1 = 2 c \beta \int_\R \phi(\beta y) \, \phi(y)\,dy 
  = \frac{2c\beta}{\sqrt{2\pi}} \int_\R
  \frac{1}{\sqrt{2\pi}}
  \exp\left( -(1+\beta^2)y^2/2 \right)\,dy.
  \]
  By the change of variable $y=x/\sqrt{1+\beta^2}$, one finally
  obtains
  \[
  \mu_1  = \frac{2c\beta}{\sqrt{2\pi}\sqrt{1+\beta^2}}
  \int_\R \frac{1}{\sqrt{2\pi}} e^{-x^2/2}\,dx
  = \frac{2c\beta}{\sqrt{2\pi} \sqrt{1+\beta^2}}.
  \qedhere
  \]
\end{proof}
\begin{rmk}
  As shown in Section \ref{s:joint} below, it is possible to provide a
  complete probabilistic characterization of the trading profit as a
  random variable, determining its density in closed form.
\end{rmk}
Recalling the definitions of $a$, $b$, $c$ and $\beta$, the
expected payoff $\mu_1$ can also be written as
\begin{equation}\label{eq:expgain}
  \mu_1 = \sqrt{2/\pi} \,
  \frac{\sigma^2_p-\sigma_{1p}}{\sqrt{\sigma^2_1+\sigma^2_p-2\sigma_{1p}}}
  = \sqrt{2/\pi} \, 
  \frac{\sigma^2_p-\sigma_1\sigma_p\rho_{1p}}%
       {\sqrt{\sigma^2_1+\sigma^2_p-2\sigma_1\sigma_p\rho_{1p}}},
\end{equation}
where $\left]-1,1\right[ \ni \rho_{1p} := \sigma_{1p}/(\sigma_1
\sigma_p)$ denotes the correlation between $\xi_1$ and $p$.  Compared
to previous (numerical) simulation studies, this closed-form solution
allows for a thorough comparative static analysis. As can be seen from
\eqref{eq:expgain}, $(\sigma_1,\sigma_r) \mapsto
\mu_1(\sigma_1,\sigma_r)$ is homogeneous of order $1$. Therefore,
without loss of generality, in Figure \ref{f:eg} we set $\sigma_p=1$
and show the effect of different levels of $\sigma_1$ and $\rho$.
\begin{figure}
\label{f:eg}
\begin{center}
\includegraphics[scale=.35]{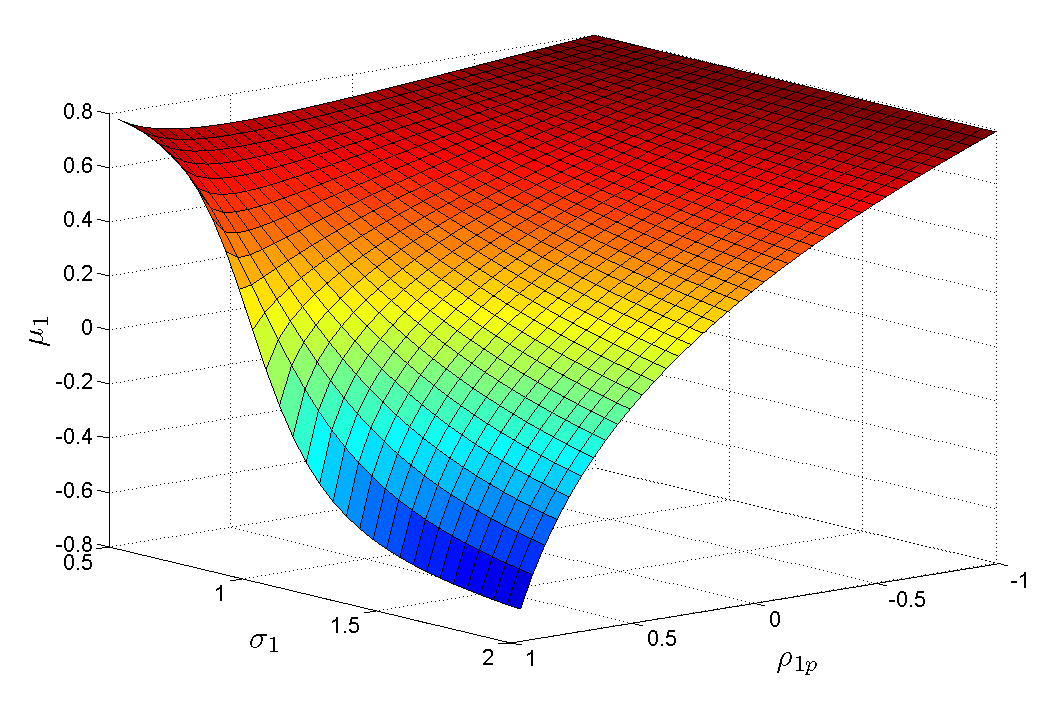}
\end{center}
\caption{Expected gains for agent 1 given $\sigma_p=1$}
\end{figure}

\subsection{Sensitivity analysis}\label{s:imp}
Let us introduce the random variable $r:=p - w_1\xi_1 = \sum_{k=2}^n
w_k\xi_k$, and define
\[
\sigma_r^2 := \operatorname{Var}(r), \qquad
\rho_{1r} := \frac{1}{\sigma_1\sigma_r} \operatorname{Cov}(\xi_1,r).
\]
In this section we analyze the impact on the expected payoff $\mu_1$
due to variations in the parameters $\sigma_1$, $\sigma_r$, and
$\rho_{1r}$. Note that one has
\begin{align*}
  \sigma_p^2 &= \operatorname{Var}(w_1\xi_1 + r) 
  = w_1^2\sigma_1^2 + \sigma_r^2 + 2w_1\sigma_{1r}
  = w_1^2\sigma_1^2 + \sigma_r^2 + 2w_1\sigma_1\sigma_r\rho_{1r},\\
  \sigma_{1p} &= \operatorname{Cov}(\xi_1,w_1\xi_1+r)
  = w_1 \sigma_1^2 + \sigma_{1r}
  = w_1 \sigma_1^2 + \sigma_1\sigma_r\rho_{1r},
\end{align*}
which implies, by \eqref{eq:expgain},
\begin{equation}     \label{eq:mu-alt}
  \mu_1 = \mu_1(\sigma_1,\sigma_r,\rho_{1r}) = 
  - \sqrt{2/\pi} \; 
  \frac{w_1(1-w_1)\sigma_1^2 - \sigma_r^2 + (1-2w_1)\rho_{1r}\sigma_r\sigma_1}%
  {\sqrt{(w_1-1)^2\sigma_1^2+\sigma_r^2-2(1-w_1)\rho_{1r}\sigma_r\sigma_1}}.
\end{equation}
\begin{prop}\label{p:sigma1}
  The following properties of the function $\sigma_1 \mapsto
  \mu_1(\sigma_1,\sigma_r,\rho_{1r})$ hold:
  \begin{itemize}
  \item[\emph{(a)}] it is strictly decreasing if $\rho_{1r} \leq 0$;
  \item[\emph{(b)}] if $\rho_{1r} > 0$, there exists $\bar{\sigma}_1>0$ such
    that it is strictly increasing on $\left]0,\bar{\sigma}_1\right]$
    and strictly decreasing on $\left[\bar{\sigma}_1,\infty\right[$.
  \end{itemize}
\end{prop}
\begin{proof} 
  (a) Let us write
  \[
  \mu_1
  = - \sqrt{2/\pi} \; \frac{a\sigma_1^2+b\sigma_1-c}%
  {\sqrt{\alpha\sigma_1^2-\beta\sigma_1+c}},
  \]
  where
  \begin{align*}
    a &:= w_1(1-w_1) > 0, &    b &:= (1-2w_1)\rho_{1r}\sigma_r,\\
    c &:= \sigma_r^2 > 0, &    \alpha &:= (1-w_1)^2 > 0,\\
    \beta &:= 2(1-w_1)\rho_{1r}\sigma_r.
  \end{align*}
  After a few tedious but straightforward calculations, one gets
  \[
  \frac{\partial\mu_1}{\partial\sigma_1}
  = - \sqrt{2/\pi} \bigl( \alpha\sigma_1^2-\beta\sigma_1+c \bigr)^{-3/2} 
  \, g(\sigma_1),
  \]
  where
  \begin{align*}
    g(x) &= a \alpha x^3
    - \frac32 a \beta x^2
    + \bigl( (2a+\alpha) c - b \beta/2 \big)x
    + (b-\beta/2)c\\
    &=: a_3x^3 + a_2x^2 + a_1x + a_0
  \end{align*}
  and
  \begin{align*}
    a_3 &= w_1(1-w_1)^3 >0,\\
    a_2 &= -3w_1(1-w_1)^2\sigma_r \rho_{1r},\\
    a_1 &= \sigma_r^2 (1-w_1)
    \bigl( (1+2\rho_{1r}^2)w_1 + 1-\rho_{1r}^2 \bigr) >0,\\
    a_0 &= -w_1 \sigma_r^3 \rho_{1r}.
  \end{align*}
  It is thus immediately seen that, if $\rho_{1r} \leq 0$, then $g(x) > 0$
  for all $x \geq 0$, which proves (a).
  \smallskip\par\noindent
  (b) We are going to show that, independently of the value of
  $\rho_{1r}>0$, the function $g$ is strictly increasing. In fact, one
  has $g'(x)=3a_3x^2 + 2a_2x + a_1$, whose discriminant $a_2^2 -
  3a_1a_3$ has the same sign of $(1+w_1)(\rho_{1r}^2-1)$, which is
  clearly negative. Therefore, $g'(x)>0$ for all $x \geq 0$, hence $g$
  is strictly increasing. In particular, since $g(0)=a_0<0$ and
  $\lim_{x\to+\infty} g(x)=+\infty$, it follows that $g$ has only one
  positive root $\bar{\sigma}_1$, with $g(x)<0$ for all $x \in
  \left[0,\bar{\sigma}_1\right[$ and $g(x)>0$ for all $x \in
  \left]\bar{\sigma}_1,+\infty\right[$, which proves (b).
\end{proof}
Therefore, if the signal is negatively correlated with the aggregated
signal of the others, then --~ceteris paribus~-- a reduction in the
volatility of the signal always increases the expected payoff. Our
model is therefore at least partially consistent with empirical
results reported in \cite{Mikhail1999}, where single market
participants have an incentive to acquire and process information in
order to improve their forecast ability. On the other hand, it should
be stressed that an improvement of the precision of the signal (i.e. a
reduction of $\sigma_1$) does not always imply a higher expected
payoff. In particular, if the signal of agent $1$ is positively
correlated with the aggregated signal of all other agents, the effect
of a decrease in $\sigma_1$ on $\mu_1$ depends, in a complex way, on
all parameters of the model.  We shall see that an analogous
non-monotonic relationship holds with respect to variations in
$\sigma_r$.

\begin{prop}\label{p:sigmar}
  The following properties of the function $\sigma_r \mapsto
  \mu_1(\sigma_1,\sigma_r,\rho_{1r})$ hold:
  \begin{itemize}
  \item[\emph{(a)}] it is strictly increasing if $\rho_{1r} \leq 0$;
  \item[\emph{(b)}] if $\rho_{1r} > 0$, there exists $\bar{\sigma}_r>0$ such
    that it is strictly decreasing on $\left]0,\bar{\sigma}_r\right]$
    and strictly increasing on $\left[\bar{\sigma}_r,\infty\right[$.
  \end{itemize}
\end{prop}
\begin{proof}
  Since $(\sigma_1,\sigma_r) \mapsto \mu_1(\sigma_1,\sigma_r)$ is
  homogeneous of order $1$, Euler's theorem yields
  \begin{equation}     \label{eq:eule}
    \sigma_1 \frac{\partial\mu_1}{\partial\sigma_1}(\sigma_1,\sigma_r)
  + \sigma_r \frac{\partial\mu_1}{\partial\sigma_r}(\sigma_1,\sigma_r)
  = \mu_1(\sigma_1,\sigma_r),
  \end{equation}
  hence
  \begin{align*}
  \frac{\partial\mu_1}{\partial\sigma_r} &= \frac{1}{\sigma_r}
  \Bigl( \mu_1 - \sigma_1 \frac{\partial\mu_1}{\partial\sigma_1} \Bigr)\\
  &= \frac{\sqrt{2/\pi}}{\sigma_r}
     \bigl( \alpha\sigma_1^2-\beta\sigma_1+c \bigr)^{-3/2}
     \Bigl( \sigma_1 g(\sigma_1) - \bigl( a\sigma_1^2 + b\sigma_1 - c\bigr)
     \bigl( \alpha\sigma_1^2 - \beta\sigma_1 + c \bigl) \Bigr)\\
  &= \frac{\sqrt{2/\pi}}{\sigma_r}
     \bigl( \alpha\sigma_1^2-\beta\sigma_1+c \bigr)^{-3/2} \Bigl(
     (-a\beta/2-b\alpha) \sigma_1^3
     + \bigl( c(a+2\alpha) + b\beta/2 \bigr) \sigma_1^2 +\\
  &\hspace{3em} -\frac32 c\beta \sigma_1 + c^2 \Bigr)\\
  &= \sqrt{2/\pi} \bigl( \alpha\sigma_1^2-\beta\sigma_1+c \bigr)^{-3/2}
     \bigl( \sigma_r^3 + a_1 \sigma_1\sigma_r^2 + a_2\sigma_1^2\sigma_r 
     + a_3 \sigma_1^3 \bigr)\\
  &= \sqrt{2/\pi} \bigl( \alpha\sigma_1^2-\beta\sigma_1+c \bigr)^{-3/2}
    h(\sigma_r),
  \end{align*}
  where $h(x):=x^3 + a_1\sigma_1x^2 + a_2\sigma_1^2x + a_3\sigma_1^3$
  and
  \begin{align*}
    a_3 &:= \frac{-a\beta/2-b\alpha}{\sigma_r} = -(1-w_1)^3 \rho_{1r},\\
    a_2 &:= \frac{c(a+2\alpha) + b\beta/2}{\sigma_r^2}
          = w_1^2(1+2\rho_{1r}^2) - 3w_1(1+\rho_{1r}^2)
             +2+\rho_{1r}^2,\\
    a_1 &:= \frac{-3c\beta/2}{\sigma_r^3} = -3(1-w_1) \rho_{1r}.
  \end{align*}
  Let us first show that $a_2>0$: in fact, looking at the above
  definition of $a_2$ as a polynomial in $w_1$, its roots are
  \[
  \frac{3(1+\rho_{1r}^2) \pm (1-\rho_{1r}^2)}{2+4\rho_{1r}^2} \geq 1.
  \]
  (a) Since $a_1, a_2, a_3>0$, Descartes' rule of signs implies that
  $h$ has no positive roots. Moreover, since $h(0)=\sigma_1^3a_3>0$,
  we conclude that $h(x)>0$ for all $x>0$,
  i.e. $\partial\mu_1/\partial\sigma_r>0$.
  \smallskip\par\noindent
  (b) Note that $h$ is strictly increasing: in fact, one has $h'(x) =
  3x^2 +2a_1\sigma_1 x + a_2\sigma_1^2$, and the discriminant of this
  polynomial is proportional to
  \[
  a_1^2-3a_2 = 3(1-\rho_{1r}^2)(1-w_1)(w_1-2)<0.
  \]
  Since $a_3<0$ implies that $h(0)<0$, we conclude that $h$ admits
  exactly one positive root $\bar{\sigma}_r$, as well as that $h$ is
  negative on $[0,\bar{\sigma}_r]$ and positive on
  $\left]\bar{\sigma}_r,+\infty\right[$.
\end{proof}

Note that the first statement of the previous Proposition simply says
that, if $\rho_{1r} \leq 0$, the expected payoff $\mu_1$ improves as
$\sigma_r$ increases, i.e. agent $1$ obtains a higher payoff as the
relative accuracy of his own signal improves. On the other hand, if
$\rho_{1r} > 0$, in analogy to Proposition \ref{p:sigma1}, the
relationship between $\mu_1$ and $\sigma_r$ is no longer monotonic and
is determined in a complex way by all parameters of the model.  This
$J$-shaped relationship between forecast precision and trading
profitability was first reported in simulation studies by
\cite{Schredelseker1984, Schredelseker2001} and also documented in
experimental works by \cite{huber2007} and \cite{Huber2008}. The
authors use a cumulative information setting which implies a positive
pairwise correlation.  In the numerical example of Section~\ref{s:num}
we adopt this setting, thus reproducing the behavior observed in the
above mentioned papers.

It should be remarked that, even though Propositions \ref{p:sigma1}
and \ref{p:sigmar} are qualitatively very similar, it is not true in
general that $\partial\mu_1/\partial\sigma_1(\sigma_1,\sigma_r)>0$
implies $\partial\mu_1/\partial\sigma_r(\sigma_1,\sigma_r)<0$ for all
$\sigma_1$, $\sigma_r$, as it can be immediately realized looking at
\eqref{eq:eule}. More precise information on the values of the
parameters determining the signs of $\mu_1$ and of its partial
derivatives can be easily obtained.

\begin{prop}
  Let $\sigma_1$, $\sigma_r$ be fixed positive constants, and
  $\alpha_+=\alpha_+(\rho_{1r})$ be defined as in \eqref{eq:alpha}
  below. The function $\rho_{1r} \mapsto
  \mu_1(\sigma_1,\sigma_r,\rho_{1r})$ is locally decreasing if
  $\sigma_r < \alpha_+ \sigma_1$ and locally increasing if $\sigma_r >
  \alpha_+ \sigma_1$.
\end{prop}
\begin{proof}
 One has, setting $\alpha:=\sigma_r/\sigma_1$,
 \begin{align*}
   \frac{\partial\mu_1}{\partial\rho_{1r}} &= \sqrt{2/\pi} \, A^{-3}
   \Bigl[ (1-w_1)(1-2w_1) \sigma_1^2 \sigma_r^2 \rho +
   w_1 \sigma_1 \sigma_r^3 -(1-w_1)^3\sigma_1^3 \sigma_r \Bigr]\\
   &= \sqrt{2/\pi} \, A^{-3} \sigma_1^3\sigma_r \Bigl[ w_1\alpha^2 +
   (1-w_1)(1-2w_1)\rho_{1r}\alpha - (1-w_1)^3 \Bigr],
 \end{align*}
 where $A$ denotes the denominator of the fraction appearing in
 \eqref{eq:mu-alt}.  Then the sign of
 $\partial\mu_1/\partial\rho_{1r}$ is equal to the sign of the
 polynomial in $\alpha$
 \[
 w_1\alpha^2 + (1-w_1)(1-2w_1)\rho_{1r}\alpha - (1-w_1)^3,
 \]
 whose roots $\alpha_-< 0 < \alpha_+$ are
 \begin{equation}     \label{eq:alpha}
 \alpha_\pm := (1-w_1)
 \frac{-(1-2w_1)\rho_{1r} \pm \sqrt{(1-2w_1)^2\rho_{1r}^2+4w_1(1-w_1)}}{2w_1}.
 \end{equation}
 In particular, if $\sigma_r/\sigma_1 < \alpha_+$, then
 $\partial\mu_1/\partial\rho_{1r}$ is negative; if $\sigma_r/\sigma_1 >
 \alpha_+$, then $\partial\mu_1/\partial\rho_{1r}$ is positive.
\end{proof}

\begin{coroll}
  Let $\sigma_1$, $\sigma_r>0$ be given, and define
  \begin{align*}
    \ell &:= \frac{1-w_1}{2w_1} \Bigl(
    \sqrt{(1-2w_1)^2+4w_1(1-w_1)} - \lvert 1-2w_1 \rvert \Bigr),\\
    u &:= \frac{1-w_1}{2w_1} \Bigl( \sqrt{(1-2w_1)^2+4w_1(1-w_1)} +
    \lvert 1-2w_1 \rvert \Bigr).
  \end{align*}
  If $\sigma_r/\sigma_1 \in [\ell,u]$, then
  \begin{itemize}
  \item[\emph{(a)}] if $w_1<1/2$, then there exists $\bar{\rho}$ such
    that $\rho_{1r} \mapsto \mu_1(\rho_{1r})$ is decreasing on $[-1,\bar{\rho}]$
    and increasing on $[\bar{\rho},1]$;
  \item[\emph{(b)}] if $w_1>1/2$, then there exists $\bar{\rho}$ such
    that $\rho_{1r} \mapsto \mu_1(\rho_{1r})$ is increasing on $[-1,\bar{\rho}]$
    and decreasing on $[\bar{\rho},1]$;
  \end{itemize}
  Otherwise, if $\sigma_r/\sigma_1 < \ell$, then $\rho_{1r} \mapsto
  \mu_1(\rho_{1r})$ is decreasing on $[-1,1]$; if $\sigma_r/\sigma_1 > u$,
  then $\rho_{1r} \mapsto \mu_1(\rho_{1r})$ is increasing on $[-1,1]$.
\end{coroll}
\begin{proof}
  Note that one has
  \[
  \frac{\partial\alpha_+}{\partial\rho_{1r}} = \frac{1-w_1}{2w_1} \, 
  \frac{-a\sqrt{a^2\rho_{1r}^2+b}+a^2\rho_{1r}}{\sqrt{a^2\rho_{1r}^2+b}},
  \]
  where
  \[
  a:=1-2w_1, \qquad b := 4w_1(1-w_1)>0,
  \]
  therefore $\partial\alpha_+/\partial\rho_{1r}>0$ if and only if
  \[
  a^2\rho_{1r} > a \sqrt{a^2\rho_{1r}^2+b}.
  \]
  Simple calculations immediately reveal that this inequality is
  always satisfied if $a=1-2w_1<0$, and never satisfied if
  $a=1-2w_1>0$. Equivalently, $\rho_{1r} \mapsto \alpha_+(\rho_{1r})$
  is increasing if $w_1>1/2$, and decreasing if $w_1<1/2$. To complete
  the proofs it is enough to observe that one has
  \begin{align*}
    \alpha_+(-1) &= \frac{1-w_1}{2w_1} \Bigl(
    (1-2w_1) + \sqrt{(1-2w_1)^2+4w_1(1-w_1)} \Bigr)\\
    \alpha_+(1) &= \frac{1-w_1}{2w_1} \Bigl( -(1-2w_1) +
    \sqrt{(1-2w_1)^2+4w_1(1-w_1)} \Bigr).  \qedhere
  \end{align*}
\end{proof}
To describe the relationship between the expect payoff of agent $1$
and the correlation coefficient $\rho_{1r}$ we can thus distinguish
three regimes: if the relative accuracy of his signal is low
($\sigma_r/\sigma_1<\ell$), then, for any $\rho_{ir} \in [-1,1]$,
agent $1$ gains from a decline in the correlation between his signal
and the aggregated signal of all other agents; if the relative
accuracy of his signal is high ($\sigma_r/\sigma_1>u$), then, for any
$\rho_{ir}\in [-1,1]$, agent 1 gains from an increase in correlation;
if the (inverse) relative accuracy $\sigma_r/\sigma_1$ falls within
the interval $[\ell,u]$, then the dependence of $\mu_1$ on $\rho_{1r}$
turns out to be non-monotonic. In the first two cases, i.e. when the
(inverse) relative accuracy $\sigma_r/\sigma_1$ falls above $u$ or
below $\ell$, the result can be heuristically motivated as follows: if
agent 1 overestimates (underestimates) the true value, then the market
price will be even higher (lower) than his signal. According to his
decision rule (buying when the signal is above the market price and
selling when the signal is below the market price), he will then take
the correct trading position.  Moreover, note that, in the third
regime, if $w_1<1/2$ the function $\rho_{1r} \mapsto \mu_1(\rho_{1r})$
has exactly one global minimum, which implies a higher expected payoff
for extreme (either positive or negative) rather than for intermediate
levels of correlation. This result hence offers a complete explanation
of the contradicting empirical observations according to which
analysts with a high reputation tend to herd \citep[see,
e.g.][]{Graham1999}, as well as to deviate more drastically from the
consensus forecast \citep[see, e.g.,][]{Lamont2002, Bernhardt2006}.

\subsection{Full Communication Equilibrium\label{s:ree}}
In Section \ref{s:imp} we consider naive investors who trade on the
basis of their individual signals. We propose a one-period model where
the true value of the firm, and therefore gains and losses of each
analyst, are revealed immediately after the trade. Of course, in the
long run the different agents will participate to the market only
under the condition that the expected payoff is not negative. This
holds also for the market marker, whose expected payoff -- by
definition of a dealer market -- is given by the negative sum of the
single $\mu_i$'s. In order to avoid a breakdown of trading as
described by the theory of lemon markets and given the zero-sum
property of the game, an equilibrium implies that the expected payoff
of each agent must be equal to zero (i.e, $\mu_i = 0$ $\forall
i=1,\ldots,n$). A natural question is whether our model, where prices
are equal to a weighted average of the single signals, see
\eqref{eq:p}, is able to accommodate such an equilibrium.

In order to keep the market alive, the market maker will have an
interest to ensure a level playing field for all participants. He
could use the observed covariance matrix $Q$ and the single signals
$\xi_i$ revealed by the conditional orders to set the price. Using the
model of \cite{Black1992} with non-informative priors, one obtains the
following choice for the price: $p^*:=\ip{Q^{-1}
  \1}{\1}^{-1}\ip{Q^{-1} \1 }{\xi}$, which corresponds in \eqref{eq:p}
to
\begin{equation}
  w^*:=\ip{Q^{-1} \1}{\1}^{-1}(Q^{-1} \1). 
  \label{eq:w}
\end{equation}
If market prices are fully revealing all available information,
then the market is said to be in a full communication equilibrium
\citep[see][]{Radner1979}. With this choice of $w$ one has, for any $k
\in \{1,2,\ldots,n\}$,
\[
\sigma_{kp} = \ip{Q_k}{w^*} = \frac{\ip{Q^{-1}Q_1}{\1}}{\ip{Q^{-1}\1}{\1}}
= \frac{\ip{e_k}{\1}}{\ip{Q^{-1}\1}{\1}} = \frac1{\ip{Q^{-1}\1}{\1}}
= \ip{Qw^*}{w^*} = \sigma_p^2,
\]
hence, by \eqref{eq:expgain}, in equilibrium the expected payoff of
all market participants is zero.  As a consequence, the expected
payoff of the market marker is also equal to zero. According to
\cite{Radner1979} any fully-revealing communication equilibrium is
also a fully-revealing rational expectation equilibrium. Let us show
that, in fact, such equilibrium is unique, in the sense of the
following Proposition.

\begin{prop}
  There exists one and only one vector $w \in \erre^n$, with
  $\sum_{k=1}^n w_k=1$, such that $\mu_i=0$ for all
  $i=1,2,\ldots,n$.
\end{prop}
\begin{proof}
  It is enough to show that there exists a unique vector $w \in
  \erre^n $ with $\sum_{k=1}^n w_k=1$, such that
  \begin{equation}     \label{eq:ww}
    \ip{Qw}{w} = \ip{Q_k}{w} \qquad \forall k=1,\ldots,n.
  \end{equation}
  Existence has already been proved by explicitly constructing a
  solution $w^*$. It is thus enough to prove uniqueness. Condition
  \eqref{eq:ww} implies $\ip{Q_1}{w} = \ip{Q_k}{w}$ for all $k>1$,
  hence $\ip{Q_1-Q_k}{w} = 0$ for all $k>1$. This in turn implies $w
  \in V^\perp$, where
  \[
  V = \operatorname{span} (Q_1-Q_2, \ldots, Q_1-Q_n),
  \]
  and $V^\perp$ stands for the orthogonal complement of $V$ in 
  $\erre^n$, so that $\erre^n = V \oplus V^\perp$. Since $Q$ is
  assumed to be non-singular, the vectors $Q_1,\ldots,Q_n$ are
  linearly independent, hence $\dim V = n-1$, which implies that $\dim
  V^\perp = 1$. Since $w^* \neq 0$, then $w^*$ is a generator of
  $V^\perp$. In particular, $0 \neq w \in V^\perp$ implies that there
  exists $\alpha \neq 0$ such that $w=\alpha w^*$. Then $1 =
  \sum_{k=1}^n w_k=\alpha \sum_{k=1}^n w^*_k = \alpha$ yields the
  uniqueness of $w^*$.
\end{proof}


\section{Density function of the payoff}
\label{s:joint}
In the previous section we have presented a closed-form solution for
the expected trading payoffs of the single agents, see
\eqref{eq:expgain}. As a matter of fact, we can give a complete
characterization of trading payoffs as random variables. In this
section we provide a closed-form expression for the density of the
payoff for each agent.  Throughout the section we adopt the notation
introduced in Lemma \ref{lm:S} and Proposition \ref{prop:22}.

We begin with an auxiliary result, which might be interesting in its
own right.
\begin{prop}     \label{prop:joint}
  Let
  \[
  F(x_1,x_2) := \P \bigl( p \leq x_1,\, \sgn(\xi_1-p) = x_2 \bigr),
  \qquad x_1 \in \erre, \; x_2 \in \{-1,1\},
  \]
  denote the joint distribution of $p$ and $\sgn(\xi_1-p)$. One has
  \[
  F(x_1,x_2;\beta) = \Phi(x_1/c) \Phi(-\beta x_1x_2/c)
  - \frac{x_2}{2\pi} \arctan(1/\beta) 
  + x_2 \sgn\beta \biggl( \frac14 
  + \int_0^{\frac{|\beta|}{c}x_1} \Phi(z/|\beta|) \phi(z)\,dz \biggr).
  \]
\end{prop}
\begin{proof}
  By Lemma \ref{lm:S}, there exist $a>0$, $b\in\erre$ and $c>0$ such
  that $p=cY$ and $\xi_1=aX+bY$ in law, where $X$ and $Y$ are
  independent standard Gaussian random variables. Setting
  $\beta=(c-b)/a$, one has
  \[
  F(x_1,x_2) = \P\bigl(Y \leq x_1/c,\; \sgn(X-\beta Y) = x_2\bigr).
  \]
  Let us record, for later use, the following obvious observation:
  \begin{align*}
    \Phi(x_1/c) &= \P(Y \leq x_1/c) = \P\bigl( Y \leq x_1/c,\,
    \sgn(X-\beta Y) = -1 \bigr)
    + \P\bigl( Y \leq x_1/c,\, \sgn(X-\beta Y) = 1 \bigr)\\
    &= F(x_1,-1) + F(x_1,1).
  \end{align*}
  Let us consider first the case $\beta>0$: one has
  \[
  F(x_1,-1) = \P\bigl( Y \leq x_1/c,\, \sgn(X-\beta Y) = -1 \bigr) =
  \P\bigl( Y \leq x_1/c,\, X<\beta Y \bigr),
  \]
  where we have used the fact that $a>0$, $c>0$. Hence
  \begin{align*}
    F(x_1,-1) &= \P\bigl(X/\beta \leq Y \leq x_1/c \bigr) \\
    &= \int_{-\infty}^{\frac{\beta}{c}x_1} \P(z/\beta \leq Y \leq
    x_1/c)
    \, \phi(z)\,dz\\
    &= \Phi(x_1/c) \Phi(\beta x_1/c)
    - \int_{-\infty}^{\frac{\beta}{c}x_1} \Phi(z/\beta) \phi(z)\,dz.
  \end{align*}
  Appealing to Lemma \ref{lm:atan}, we end up with
  \[
  F(x_1,-1) = \Phi(x_1/c) \Phi(\beta x_1/c)
  - \frac14 + \frac{1}{2\pi} \arctan(1/\beta) - \int_0^{\frac{\beta}{c}x_1}
  \Phi(z/\beta) \phi(z)\,dz.
  \]
  The expression for $F(x_1,1)$ is obtained as follows:
  \begin{align*}
    F(x_1,1) &= \Phi(x_1/c) - F(x_1,-1)\\
    &= \Phi(x_1/c) \bigl( 1 - \Phi(\beta x_1/c) \bigr) +\frac14 -
    \frac{1}{2\pi} \arctan(1/\beta)
    + \int_0^{\frac{\beta}{c}x_1} \Phi(z/\beta) \phi(z)\,dz\\
    &= \Phi(x_1/c) \Phi(-\beta x_1/c) + \frac14 - \frac{1}{2\pi}
    \arctan(1/\beta)
    + \int_0^{\frac{\beta}{c}x_1} \Phi(z/\beta) \phi(z)\,dz\\
  \end{align*}
  Let us now consider the case $\beta<0$: by a reasoning completely
  analogous to the one used above, we can write
  \begin{align*}
    F(x_1,1) &= \P\bigl(Y \leq x_1/c,\; X/\beta \leq Y\bigr)
    = \P(X/\beta \leq Y \leq x_1/c)\\
    &= \int_{\frac{\beta}{c}x_1}^\infty
       \P(z/\beta \leq Y \leq x_1/c)\,\phi(z)\,dz\\
    &= \Phi(x_1/c) \Phi(-\beta x_1/c)
    - \int_{-\infty}^{-\frac{\beta}{c}x_1} \Phi(-z/\beta)\phi(z)\,dz\\
    &= \Phi(x_1/c) \Phi(-\beta x_1/c) - \frac14 - \frac{1}{2\pi}
    \arctan(1/\beta)
    - \int_0^{\frac{|\beta|}{c}x_1} \Phi(z/|\beta|) \phi(z)\,dz,
  \end{align*}
  hence also
  \begin{align*}
    F(x_1,-1) &= \Phi(x_1/c) - F(x_1,1)\\
    &= \Phi(x_1/c) \bigl(1 - \Phi(-\beta x_1/c) \bigr) + \frac14 +
    \frac{1}{2\pi} \arctan(1/\beta)
    + \int_0^{\frac{|\beta|}{c}x_1} \Phi(z/|\beta|) \phi(z)\,dz\\
    &= \Phi(x_1/c) \Phi(\beta x_1/c) + \frac14 + \frac{1}{2\pi}
    \arctan(1/\beta) + \int_0^{\frac{|\beta|}{c}x_1} \Phi(z/|\beta|)
    \phi(z)\,dz.
  \end{align*}
  We may thus write
  \begin{align*}
    F(x_1,-1) &= \Phi(x_1/c) \Phi(\beta x_1/c) + \frac{1}{2\pi}
    \arctan(1/\beta) - \sgn\beta \left( \frac14 +
      \int_0^{\frac{|\beta|}{c}x_1}
      \Phi(z/\beta) \phi(z)\,dz\right),\\
    F(x_1,1) &= \Phi(x_1/c) \Phi(-\beta x_1/c) - \frac{1}{2\pi}
    \arctan(1/\beta) + \sgn\beta \left(\frac14 +
      \int_0^{\frac{|\beta|}{c}x_1} \Phi(z/\beta) \phi(z)\,dz\right),
  \end{align*}
  that are equivalent to the claim.
\end{proof}

\begin{rmk}
  The joint distribution of $p$ and $\sgn(\xi_1-p)$ can alternatively
  be expressed in terms of the bivariate Gaussian law. Consider, for
  instance, the case $\beta>0$ and $x_2=-1$: Lemma \ref{lm:phi2}
  yields
  \begin{align*}
    F(x_1,-1) &= \Phi(x_1/c) \Phi(\beta x_1/c)
    - \int_{-\infty}^{\frac{\beta}{c}x_1} \Phi(z/\beta) \phi(z)\,dz\\
    &= \Phi(x_1/c) \Phi(\beta x_1/c) - \Phi_2\left(\beta x_1/c,0;
      -\frac{1}{\sqrt{1+\beta^2}}\right),
  \end{align*}
  where $\Phi_2(\cdot,\cdot;\rho)$ denotes the distribution function
  of a bivariate Gaussian random variable with correlation coefficient
  $\rho$.
\end{rmk}

The following Proposition is the main result of this section and of
the whole paper.
\begin{prop}
  Let $M=p \sgn(\xi_1-p)$ denote the negative payoff of agent $1$. The
  random variable $M$ has a (smooth) density $f_M$ given by
  \[
  f_M(z) = \frac2c \phi(z/c) \Phi(-\beta z/c).
  \]
\end{prop}
\begin{proof}
  Let us first compute the distribution $F_M(z):=\P(M \leq z)$ of the
  random variable $M$. One has
  \begin{align*}
    F_M(z) &= \P( p \sgn(\xi_1-p) \leq z)\\
    &= \P\bigl( p \sgn(\xi_1-p) \leq z,\; \sgn(\xi_1-p)=-1 \bigr)
    + \P\bigl( p \sgn(\xi_1-p) \leq z,\; \sgn(\xi_1-p)=1 \bigr)\\
    &= \P\bigl( p \geq -z,\; \sgn(\xi_1-p)=-1 \bigr)
    + \P\bigl( p \leq z,\; \sgn(\xi_1-p)=1 \bigr)\\
    &= \P(\sgn(\xi_1-p)=-1) - F(-z,-1) + F(z,1),
  \end{align*}
  where $F$ denotes the joint distribution of $p$ and $\sgn(\xi_1-p)$,
  and
  \[
  \P(\sgn(\xi_1-p)=-1) = \P(X \leq \beta Y) = 1/2,
  \]
  which follows immediately because $X$ and $Y$ are independent,
  symmetric, and have a continuous distribution: $\P(X \leq \beta Y) +
  \P(X \geq \beta Y) = 1$, but $\P(X \geq \beta Y) = \P(-X \geq -\beta
  Y) = \P(X \leq \beta Y)$.

  Thanks to Proposition \ref{prop:joint}, one has
  \[
  F_M(z) = C_\beta + \bigl(2\Phi(z/c)-1\bigr) \Phi(-\beta z/c)
  + \sgn\beta \left( \int_0^{-\frac{|\beta|}{c}z} \Phi(y/|\beta|) \phi(y)\,dy
  + \int_0^{\frac{|\beta|}{c}z} \Phi(y/|\beta|) \phi(y)\,dy \right),
  \]
  where $C_\beta$ denotes a constant that depends only on $\beta$,
  hence
  \begin{align*}
    f_M(z) = \frac{d}{dz} F_M(z) &= \frac2c \phi(z/c) \Phi(-\beta z/c)
    -\frac{\beta}{c}\phi(\beta z/c) \bigl(2\Phi(z/c)-1\bigr)\\
    &\qquad + \sgn\beta \, \Bigl( \frac{|\beta|}{c}
    \Phi(z/c)\phi(|\beta|z/c)
    - \frac{|\beta|}{c} \Phi(-z/c)\phi(-|\beta|z/c) \Bigr)\\
    &= \frac2c \phi(z/c) \Phi(-\beta z/c) + \frac{\beta}{c}\phi(\beta
    z/c)
    -2\frac{\beta}{c} \Phi(z/c) \phi(\beta z/c)\\
    &\qquad + \frac{\beta}{c} \Phi(z/c) \phi(\beta z/c)
    - \frac{\beta}{c} \Phi(-z/c) \phi(\beta z/c)\\
    &= \frac2c \phi(z/c) \Phi(-\beta z/c) + \frac{\beta}{c}\phi(\beta
    z/c)
    - \frac{\beta}{c}\phi(\beta z/c) \bigl( \Phi(z/c) + \Phi(-z/c) \bigr)\\
    &= \frac2c \phi(z/c) \Phi(-\beta z/c).
    \qedhere
  \end{align*}
\end{proof}

It is immediately seen that the density of the random variable $-M$,
which represents the trading payoff of agent $1$, is given by
\[
f_{-M}(z) = f_M(-z) = \frac2c \phi(z/c) \Phi(\beta z/c).
\]

\subsection{Higher moments}
\label{s:mom}
Given the complete probabilistic characterization of the (negative)
trading profit $M$ just obtained, one can compute any moment of $M$
simply integrating against the density $f_M$. It is easier, however,
to proceed differently.

For an odd number $2k-1$, $k\in\mathbb{N}$, we define $(2k-1)!! :=
\prod_{j=1}^k (2j-1)$, and set, by convention, $(-1)!!:=1$.
\begin{prop}
  Let $k \in \mathbb{N}$. One has
  \begin{align*}
  \E M^{2k}   &= (2k-1)!! \, \ip{Qw}{w}^k,\\
  \E M^{2k-1} &= -\frac{2^k (k-1)! \, c^{2k-1} \beta}{\sqrt{2\pi(1+\beta^2)}} \,
  \sum_{j=0}^{k-1} \frac{(2j-1)!!}{(2j)!!} \, \frac{1}{(1+\beta^2)^j}.
  \end{align*}
\end{prop}
\begin{proof}
  Since $M^{2k} = p^{2k}$, and $p=\ip{w}{\xi}$ is a centered Gaussian
  random variable with variance $\ip{Qw}{w}$, standard formulas for
  moments of Gaussian laws give
  \[
  \E M^{2k} = \E p^{2k} = \E \ip{w}{\xi}^{2k}
  = (2k-1)!! \, \ip{Qw}{w}^k.
  \]
  Moreover, by Lemma \ref{lm:S}, one has
  \begin{align*}
    \E M^{2k+1} &= \E p^{2k+1}\sgn(\xi_1-p)
    = c^{2k+1} \E Y^{2k+1} \sgn(X-\beta Y)\\
    &= c^{2k+1} \int_\erre y^{2k+1} \E\left[\sgn(X-\beta y)\right] \, \phi(y)\,dy,
  \end{align*}
  where $\E\left[\sgn(X-\beta y)\right] = 1-2\Phi(\beta y)$, hence
  \begin{equation}     \label{eq:dispari}
  \E M^{2k+1} = -2c^{2k+1} \int_\erre y^{2k+1} \Phi(\beta y)
  \,\phi(y)\,dy.
  \end{equation}
  Note that, setting $v(y)=y^{2k}\phi(y)$, one has
  \[
  v'(y) = \bigl( 2k y^{2k-1} - y^{2k+1}\bigr) \phi(y),
  \]
  hence, integrating by parts in \eqref{eq:dispari},
  \begin{equation}     \label{eq:ric}
  \begin{split}
    \int_\erre \Phi(\beta y) \, y^{2k+1} \phi(y)\,dy &= 
    \int_\erre \Phi(\beta y) 
        \bigl( y^{2k+1} - 2k y^{2k-1} \bigr)\phi(y)\,dy\\
    &\qquad + 2k \int_\erre \Phi(\beta y) \, y^{2k-1}\phi(y)\,dy\\
    &= \beta \int_\erre y^{2k} \phi(\beta y) \phi(y)\,dy 
       + 2k \int_\erre \Phi(\beta y) \, y^{2k-1}\phi(y)\,dy,
  \end{split}
  \end{equation}
  where, setting $\sigma_\beta := (1+\beta^2)^{-1/2}$ for convenience
  of notation, and recalling again standard formulas for moments of
  Gaussian laws,
  \begin{align*}
    \int_\erre y^{2k} \phi(\beta y) \phi(y)\,dy &= 
    \bigl( 2\pi(1+\beta^2) \bigr)^{-1/2}
    \frac{1}{\sigma_\beta \sqrt{2\pi}} 
    \int_\erre y^{2k} e^{-y^2/(2\sigma_\beta^2)}\,dy\\  
    &= \bigl( 2\pi(1+\beta^2) \bigr)^{-1/2} \, (2k-1)!! \,
    (1+\beta^2)^{-k}.
  \end{align*}
  Therefore, setting
  \begin{align*}
    f_k &:= \int_\erre \Phi(\beta y) \, y^{2k-1}\phi(y)\,dy,\\
    a_k &:= \frac{\beta}{\sqrt{2\pi(1+\beta^2)}} \, (2k-1)!! \,
    (1+\beta^2)^{-k},
  \end{align*}
  \eqref{eq:ric} can be written as
  \[
  f_{k+1} - 2k f_k = a_k, \qquad
  f_1=\frac{\beta}{\sqrt{2\pi(1+\beta^2)}}.
  \]
  Dividing both sides of this difference equation by $2^k k!$, we are
  left with
  \[
  \frac{f_{k+1}}{2^k\,k!} - \frac{f_k}{2^{k-1}\,(k-1)!}
  = \frac{a_k}{2^k\,k!},
  \]
  hence, setting
  \[
  g_k := \frac{f_k}{2^{k-1}\,(k-1)!},
  \qquad b_k := \frac{a_k}{2^k\,k!},
  \]
  the previous difference equation is equivalent to
  \[
  g_{k+1}-g_k = b_k, \qquad g_1 = \frac{\beta}{\sqrt{2\pi(1+\beta^2)}},
  \]
  which can be easily solved, writing the telescoping sum
  \begin{align*}
    g_k - g_1 &= (g_k - g_{k-1}) + (g_{k-1} - g_{k-2}) + \cdots
    + (g_2 - g_1)\\
    &= b_{k-1} + b_{k-2} + \cdots + b_1,
  \end{align*}
  which yields
  \begin{align*}
    f_k &= 2^{k-1} (k-1)! \, g_k
    = 2^{k-1} (k-1)! \, \biggl( g_1 + \sum_{j=1}^{k-1} b_j \biggl)\\
    &= 2^{k-1} (k-1)! \, \biggl( \frac{\beta}{\sqrt{2\pi(1+\beta^2)}}
    + \frac{\beta}{\sqrt{2\pi(1+\beta^2)}} \sum_{j=1}^{k-1}
      \frac{(2j-1)!!}{2^jj!} \, \frac{1}{(1+\beta^2)^j} \biggr)\\
    &= \frac{2^{k-1} (k-1)! \, \beta}{\sqrt{2\pi(1+\beta^2)}} \,
       \sum_{j=0}^{k-1} \frac{(2j-1)!!}{(2j)!!} \, \frac{1}{(1+\beta^2)^j}.
  \end{align*}
  By definition of $f_k$ and \eqref{eq:dispari} one finally obtains
  \[
  \E M^{2k-1} = -2c^{2k-1} f_k 
  = -\frac{2^k (k-1)! \, c^{2k-1} \beta}{\sqrt{2\pi(1+\beta^2)}} \,
  \sum_{j=0}^{k-1} \frac{(2j-1)!!}{(2j)!!} \, \frac{1}{(1+\beta^2)^j}.
  \qedhere
  \]
\end{proof}
It is immediately seen that the moments of the payoff $-M$ of even
degree coincide with those of $M$, while the moments of odd degree are
equal in absolute value, but with opposite sign. Moreover, centered
moments can easily be obtained by the non-central one just derived.

\section{Numerical example}\label{s:num}
In this section we illustrate our main results with a numerical
example. We assume that four agents, labeled by $i \in
\{1,\ldots,4\}$, participate to the market. The covariance matrix $Q$
of their signals is displayed in Table~\ref{t:signal} below. The
analysts are labeled according to their forecast precision (measured
by the standard deviation $\sigma_i$ of their signals), in reverse
order. The market maker assigns equal weight ($w_i=0.25)$ to each
signal. The numerical values chosen here try to mimic models with a
cumulative information structure \citep[see][]{Schredelseker1984,
  Schredelseker2001}, where analysts with an intermediate accuracy
implicitly face the highest correlation (note that the signal of agent
$1$ is less correlated with other agents' signals than the signal of
agent $2$).

\begin{table}[htb]
\centering
\begin{tabular}{c|c|cccc|cccc}
\hline
         $i$ &         $\sigma_i$ &                        \multicolumn{ 4}{|c}{$C$} &                           \multicolumn{ 4}{|c}{$Q$} \\
\hline
         1 &        1.3 &          1 &        0.9 &        0.6 &        0.3 &      1.690 &      1.404 &      0.858 &      0.390 \\

         2 &        1.2 &        0.9 &          1 &        0.8 &        0.6 &      1.404 &      1.440 &      1.056 &      0.720 \\

         3 &        1.1 &        0.6 &        0.8 &          1 &        0.7 &      0.858 &      1.056 &      1.210 &      0.770 \\

         4 &          1 &        0.3 &        0.6 &        0.7 &          1 &      0.390 &      0.720 &      0.770 &      1.000 \\
\hline
\end{tabular} 
\caption{Signal structure of the numerical example: $C$ and $Q$ denote the correlation matrix and the covariance matrix, respectively.} 
\label{t:signal}
\end{table}

In the following we compare the payoffs of all analysts. We use the
formulas of the previous section to calculate the first four central
moments.\footnote{Note that central moments can be immediately
  obtained from the non-central moments of Section \ref{s:mom}.}
Table~\ref{t:mom} shows that, although analyst $1$ has a less accurate
signal than analyst $2$ (with a resulting higher variance and kurtosis
in the payoffs), his expected payoff is better. Obviously, analyst
$2$ suffers from his high correlation with all other market
participants, which is also reflected in a more negative skewness.

Therefore, our model reproduces the non-monotone relationship between
forecast precision and trading profitability observed in
\cite{Schredelseker1984, Schredelseker2001}.
\begin{table}[t]
\centering
\begin{tabular}{c|cccc}
\hline
         $i$ & $\mu_i$ & $\sigma^2_i$ & $\varsigma_i$ & $\kappa_i$ \\
\hline
         1 &     -0.115 &      0.970 &     -0.001 &      2.825 \\

         2 &     -0.406 &      0.819 &     -0.029 &      2.018 \\
         3     & 0.016 & 0.983 & 0.000 & 2.900 \\
		4     & 0.285 & 0.902 & 0.010 & 2.444 \\

\hline
\end{tabular}  
\caption{Higher moments of the payoff for analysts $i=1,\ldots,4$: $\varsigma_i$ and $\kappa_i$ denote the skewness and kurtosis of agent $i$, respectively.}
\label{t:mom}
\end{table}
Figure~\ref{f:density} shows the density function of the payoff for
all analysts.
\begin{figure}
\centering
\includegraphics[scale=.6]{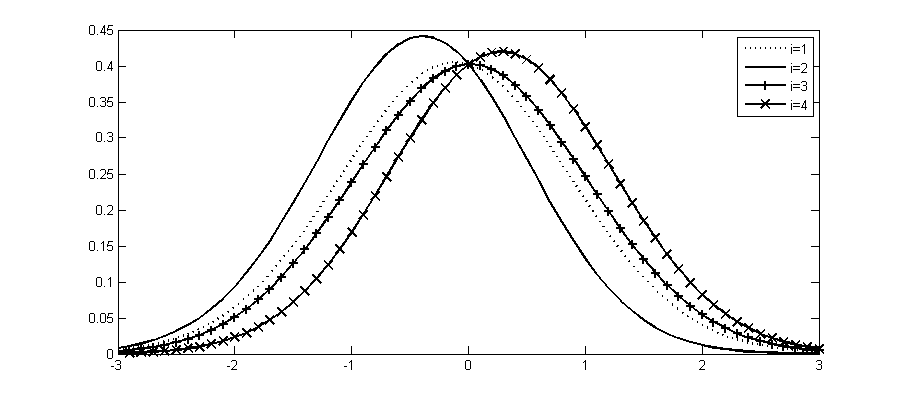} 
\caption{Density function for the payoffs of analysts $i=1,\ldots,4$.}
\label{f:density}
\end{figure}
Figure~\ref{f:mu1} shows the dependence of the expected payoff of
agent $1$ on the accuracy of his own signal (as measured by the
standard deviation $\sigma_1$) and on the correlation between his
signal $\xi_1$ and the aggregated signal of all other agents $r$. The
value of $\sigma_r$ is equal to $0.739$. The non-monotonic
relationship between $\mu_1$ and $\sigma_1$ is clearly displayed for
some positive values of $\rho_{1r}$. Moreover, while for large (small)
values of $\sigma_1$ the expected payoff is decreasing (increasing) as
$\rho_{1r}$ increases, for intermediate values of $\sigma_1$ a
non-monotonic behaviour can be observed. Given $w_1=0.25<1/2$, $\mu_1$
is first decreasing and then increasing in $\rho_{1r}$.
\begin{figure}
\centering
\includegraphics[scale=.4]{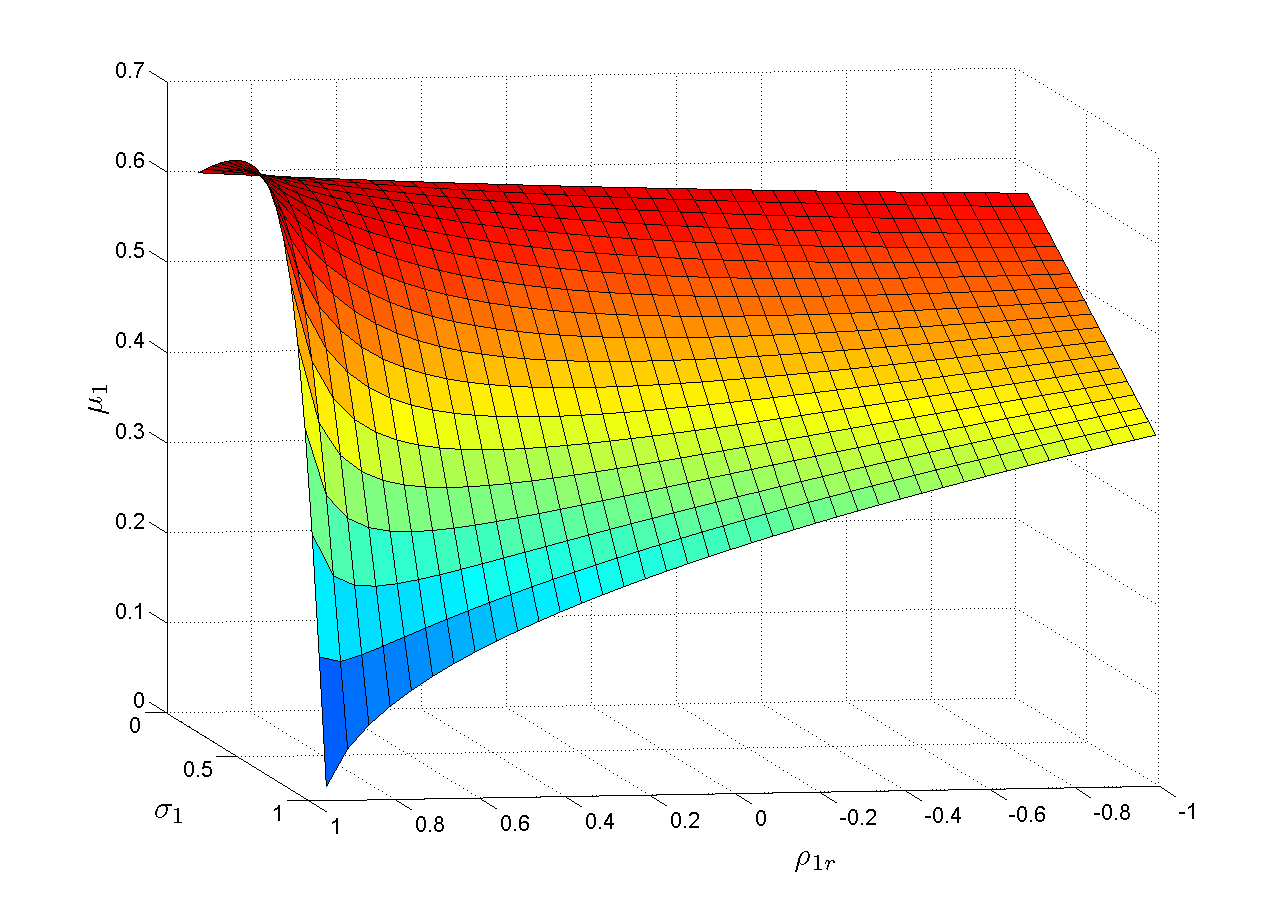} 
\caption{Expected payoff of analysts $1$ as function of $\sigma_1$ and
  $\rho_{1r}$, with $\sigma_r=0.739$.}
\label{f:mu1}
\end{figure}

\section{Conclusion}\label{s:con}
We propose a model in which analysts differ in the precision of their
signals as well as in the correlation between each other. We provide a
complete characterization of the trading payoff of each agent,
obtaining its probability density function in closed form, as well as
explicit expressions for its moments of all orders. Such precise
description allows us to perform a detailed sensitivity analysis, with
the following implications: If the forecasts of an analyst are
negatively correlated with the aggregated forecasts of the others,
then he always takes advantage of improving his forecast skills. If
the correlation between the forecasts is positive, then a
non-monotonic relationship between forecast precision and trading
profitability exists. Furthermore, the impact of correlation on the
expected payoff of an analyst depends on the relative accuracy of his
signal: if the relative accuracy is below (above) a certain threshold,
then he suffers (benefits) from an increasing correlation, while for
intermediate levels of relative accuracy the relationship is
non-monotonic.

This is the first time, to the best of our knowledge, that an
analytical model is proposed, which fully explains the non-trivial
interplay between forecast accuracy and trading performance. Our model
recovers as special cases previous (partial) results from simulation
studies, and is able to explain the different -- sometimes
contradicting -- empirical observations reported in the literature.
Finally, we provide strong evidence that empirical studies on the
relationship between forecast precision and trading profitability need
to take into account the correlation structure of the forecasts.

\appendix
\section{Appendix}
\begin{lemma}     \label{lm:atan}
  Let $a \in \erre$. One has
  \begin{align*}
    \int_0^\infty \Phi(ax)\,\phi(x)\,dx &= \frac14 + \frac{1}{2\pi}\arctan a,\\
    \int_{-\infty}^0 \Phi(ax)\,\phi(x)\,dx &= \frac14 -
    \frac{1}{2\pi}\arctan a.
  \end{align*}
\end{lemma}
\begin{proof}
  For $a>0$ one has
  \[
  \Phi(ax) = \frac12 + \int_0^{ax} \phi(y)\,dy 
  \qquad \forall x \geq 0,
  \]
  hence
  \[
  \int_0^\infty \Phi(ax)\,\phi(x)\,dx =
  \int_0^\infty \left(\frac12 + \int_0^{ax} \phi(y)\,dy \right)
  \phi(x)\,dx = \frac14 + \int_{D_a} \phi_2(x,y)\,dx\,dy,
  \]
  where $\erre^2 \supset D_a := \bigl\{ 0 \leq x < \infty, \; 0 \leq y
  \leq ax\bigr\}$ and $\phi_2$ stands for the density of the standard
  Gaussian measure on $\erre^2$. Since $D_a$ is a cone of $\erre^2$
  with vertex at the origin and aperture equal to $\arctan a$, taking
  the rotational invariance of $\phi_2$ into account (or,
  equivalently, passing to radial coordinates), we have
  \[
  \int_{D_a} \phi_2(x,y)\,dx\,dy = \frac{\arctan a}{2\pi}.
  \]
  Let us now assume $a<0$. Then $\Phi(ax) = 1-\Phi(|a|x)$ for all $x
  \geq 0$, hence, recalling that $\arctan$ is odd,
  \[
  \int_0^\infty \Phi(ax)\,\phi(x)\,dx = \frac12 
  - \int_0^\infty \Phi(|a|x)\,\phi(x)\,dx = \frac12 - \frac14 
  - \frac{\arctan |a|}{2\pi} = \frac14 + \frac{\arctan a}{2\pi}.
  \]
  The first identity is thus proved. The second follows immediately:
  \[
  \int_{-\infty}^0 \Phi(ax)\,\phi(x)\,dx = 
  \int_0^\infty \Phi(-ax)\,\phi(x)\,dx =
  \int_0^\infty \bigl(1-\Phi(ax)\bigr)\,\phi(x)\,dx =
  \frac12 - \frac14 - \frac{1}{2\pi}\arctan a.
  \qedhere
  \]
\end{proof}

\begin{lemma}     \label{lm:phi2}
  One has
  \[
  \int_{-\infty}^y \Phi(bx) \phi(x)\,dx =
  \Phi_2\bigl(y,0;-b(1+b^2)^{-1/2}\bigr)
  \]
\end{lemma}
\begin{proof}
  It is enough to write
  \[
  \Phi(bx) = \int_{-\infty}^{bx} \phi(z)\,dz =
  \int_{-\infty}^0 \phi(z+bx)\,dz = 
  \sqrt{b^2+1} \int_{-\infty}^0 \phi(z\sqrt{b^2+1}+bx)\,dz,
  \]
  which implies
  \[
  \int_{-\infty}^y \Phi(bx) \phi(x)\,dx =
  \frac{\sqrt{b^2+1}}{2\pi} \, \int_{\Xi_y} \exp\Bigl(
  - \frac12\bigl( (b^2+1)x^2 + 2b(b^2+1)^{1/2}xz + (b^2+1)z^2 \bigr) 
  \Bigr)\,dx\,dz,
  \]
  where $\erre^2 \supset \Xi_y:=\left]-\infty,y\right] \times
  \left]-\infty,0\right]$.  Writing
  \begin{align*}
  (b^2+1)x^2 + 2b(b^2+1)^{1/2}xz + (b^2+1)z^2 &= 
  (b^2+1) \left( x^2 + z^2 - 2 \frac{-b}{\sqrt{b^2+1}}xz \right)\\
  &= \frac{1}{1-\rho^2} \bigl( x^2 - 2\rho xz + z^2 \bigr),
  \end{align*}
  with $\displaystyle \rho:= -\frac{b}{(b^2+1)^{1/2}}$, we are left
  with
  \[
  \int_{-\infty}^y \Phi(bx) \phi(x)\,dx = \int_{\Xi_y} f(x,z;R)\,dx\,dz,
  \]
  where $f(\cdot,\cdot;R)$ denotes the density of $N_{\erre^2}(0,R)$, with
  \[
  R =
  \begin{bmatrix}
    1 & \rho\\
    \rho & 1
  \end{bmatrix} =
  \begin{bmatrix}
    1 & \displaystyle -\frac{b}{\sqrt{b^2+1}}\\
    \displaystyle -\frac{b}{\sqrt{b^2+1}} & 1
  \end{bmatrix}.
  \qedhere
  \]
\end{proof}

\bibliographystyle{chicago}
\bibliography{CorrError}

\end{document}